\documentclass[11pt]{article}
\usepackage[a4paper]{geometry}
\usepackage{t1enc}
\usepackage{url}
\usepackage{graphicx}
\usepackage{amsmath,bm}
\usepackage{booktabs}
\usepackage{xcolor}

\newtheorem{defi}{Definition}
\newtheorem{thm}{Theorem}

\newcommand{\adj}{\,\text{\textemdash}\,}
\def\..{\,\mathpunct{\ldotp\ldotp}} 

\newcommand{\adjugate}[1]{{\operatorname{adj}}({#1})}

\newcommand{\lift}[2]{\widetilde{#1}^{#2}}

\newcounter{noqed}
\newcommand{\qed}{ \ifmmode\mbox{
}\fi\rule[-.05em]{.3em}{.7em}\setcounter{noqed}{0}}
\newenvironment{proof}[1][{}]{\noindent{\bf Proof#1.
}\setcounter{noqed}{1}}{\ifnum\value{noqed}=1\qed\fi\par\medskip}

\renewcommand{\phi}{\varphi}
\renewcommand{\epsilon}{\varepsilon}
\renewcommand{\theta}{\vartheta}

\newcommand{\pre}{\operatorname{pre}}
\newcommand{\post}{\operatorname{post}}

\title{Monotonicity in Undirected Networks}

\author{Paolo Boldi\thanks{The authors have been supported by the FASTEN EU Project, H2020-ICT-2018-2020 (GA 825328).}\qquad Flavio Furia\qquad Sebastiano Vigna\\Universit\`a degli Studi di Milano, Dipartimento di Informatica}

\begin{document}
\bibliographystyle{plain}

\maketitle

\begin{abstract}
Is it always beneficial to create a new relationship (have a new follower/friend) in a social network?
This question can be formally stated as a property of the centrality measure
that defines the importance of the actors of the network.
\emph{Score monotonicity} means that adding an arc increases the centrality score of the
target of the arc; \emph{rank monotonicity} means that adding an arc improves the
importance of the target of the arc relatively to the remaining nodes.
It is known that most centralities are both score and rank monotone on directed,
strongly connected graphs.
In this paper, we study the problem of score and rank monotonicity for classical
centrality measures in the case of \emph{undirected} networks: in this case, we
require that score, or relative importance, improve at both endpoints of the new
edge.
We show that, surprisingly, the situation in the undirected case is very
different, and in particular that closeness, harmonic centrality, betweenness,
eigenvector centrality, Seeley's index, Katz's index, and PageRank are not rank
monotone; betweenness and PageRank are not even score monotone.
In other words, while it is always a good thing to get a new follower, it is not
always beneficial to get a new friend.
\end{abstract}

\maketitle

\section{Introduction}

The study of centrality in networks goes back to the late forties. Since then,
several measures of centrality with different properties have been proposed---see~\cite{BoVAC} for a survey.
To sort out which measures are more apt for a specific application, one can try
to classify them through some axioms that they might satisfy or not. 

In~\cite{BoVAC,BLVRMCM}, two of the authors have studied in particular
\emph{score monotonicity} and \emph{rank monotonicity} on directed graphs.
The first property says that when an arc $x\to y$ is added to the graph, the
score of $y$ strictly increases~\cite{SabCIG}.
Rank monotonicity~\cite{CDKLEAA} states that after adding an
arc $x\to y$ all nodes with a score smaller than (or equal to) $y$ have still a
score smaller than (or equal to) $y$.
Score and rank monotonicity complement themselves: score monotonicity tells us
that ``something good happens''; rank monotonicity that ``nothing bad happens''.

In some way, both axioms aim at answering the following question: is it always
worth it for a node in a directed social network (say, Twitter) to have a new
incoming arc (in Twitter parlance, a new follower)? The two monotonicity axioms
introduced above have a different interpretation of what ``worth'' means.
``Score monotonicity'' interprets it simply as an increase of score: if you get
a new follower, does your score always increase? ``Rank monotonicity''
interprets it with respect to the score of other nodes: if you get a new
follower, do you still dominate (have a larger score than) the same nodes you used
to dominate before, and possibly more? As we said, for most notions of
importance (i.e., centrality measures) the answer to both questions is ``yes'',
under very mild assumptions~\cite{BLVRMCM}.
 
Once we move to undirected graphs, however, previous definitions and results are
no longer applicable. Thus, in this paper, we aim at answering a subtly different question:
is it always worth it for an actor in an \emph{undirected} social network (say,
Facebook) to have a new friend? Again, ``worth'' can be taken to refer
to the score or to the rank.
In this paper, we propose more precise definitions that are natural extensions
of score and rank monotonicity to the undirected case, and prove results about
classical centrality measures: closeness~\cite{BavMMGS}, harmonic centrality~\cite{BeaIIC}, betweenness~\cite{AntRG,FreSMCBB}, and
four variants of spectral ranking~\cite{VigSR}---eigenvector
centrality~\cite{LanZRWT,BerTGA}, Katz's index~\cite{KatNSIDSA}, Seeley's
index~\cite{SeeNRI}, and PageRank~\cite{PBMPCR}.

As we will see, while in some cases we can witness some score increase,
except for Seeley's index \emph{none of the centrality measures we consider
is rank monotone}. This is somehow surprising, and will yield some reflection.

Note that adding a single edge to an undirected graph is equivalent to adding
\emph{two} opposite arcs in a directed graph, which may suggest why the
situation is so different, at least from the mathematical viewpoint.
Understanding under which conditions a centrality measure does not satisfy 
an axiom will be a theme that we
will try to pursue in the course of the discussion.

We provide classes of counterexamples of arbitrary size; moreover, we
always provide both a counterexample in which the loss of rank happens in the
less important node of the new edge and a counterexample in which the loss of rank happens in the
more important node of the new edge. In this way, we will show that it is impossible  
for the two actors in the social network creating the new edge to predict
whether the edge will be beneficial \emph{even knowing their relative importance}.
The results obtained in this paper are resumed in Table~\ref{tab:summ}. 
 
To prove general results in the case of spectral rankings, 
we exploit the connection between spectral rankings and
graph fibrations~\cite{BoVGF,BLSGFGIP}, which makes us able to reduce
computations on graphs with a variable number of nodes to similar
computations on graphs with a fixed number of nodes.
This approach to proofs, which we believe is of independent interest, makes it
possible to use analytic techniques to control the values assumed by
eigenvector centrality, Katz's index, and PageRank.

We conclude the paper with some anecdotal evidence from a medium-sized real-world
network, showing that violations of monotonicity do happen also in practice.

With respect to the conference paper~\cite{BFVSRMUN}, all results on geometric
centralities and betweenness are new, as well as all general results on eigenvector
centrality, and all results about Katz's index. The second PageRank counterexample
is also new. All results about demotion, for all centralities, are also new. The proofs
for the first PageRank example have been significantly simplified.   

Most of the computations in this paper (in particular, the manipulation of complex rational functions)
have been performed using Sage~\cite{Sage}. All our Sage worksheets are available at
\url{https://github.com/vigna/monotonicity}, and will be badged on the Zenodo platform after the reviewing process.

\begin{table}
\centering
\begin{tabular}{l||c|c}
& score monotonicity & rank monotonicity\\\hline
Closeness & yes & no\\
Harmonic centrality & yes & no\\
Betweenness & no & no\\
Eigenvector centrality & no & no\\
Seeley's index & yes & yes\\
Katz's index & yes & no\\
PageRank & no & no\\
 \end{tabular}
 \vspace*{3mm}
\caption{\label{tab:summ}Summary of the results of this paper for the case of \emph{connected undirected graphs}. 
For comparison, recall that~\cite{BLVRMCM} all the 
centrality measures listed are both score and rank monotone on \emph{strongly connected directed graphs},
with the only exception of betweenness that is neither.}
\end{table}

\section{Graph-theoretical preliminaries}
\label{sec:defs}

While we will focus on simple undirected graphs, we are going to make use of
some proof techniques that require handling more general types of graphs.

A \emph{(directed multi)graph} $G$ is defined by a set $N_G$ of nodes, a set
$A_G$ of arcs, and by two functions $s_G,t_G:A_G\to N_G$ that specify the
source and the target of each arc; a \emph{loop} is an arc with the same source and target;
the main difference between this definition and the standard definition of a
directed graph is that we allow for the presence of multiple arcs between a pair of nodes.
When we do not need to distinguish between multiple arcs, we write $x\to y$ to denote
an arc with source $x$ and target $y$. 



Since we do not need to discriminate between graphs that only differ because of
node names, we will often assume that $N_G=\{\,0,1,\dots,n_G-1\,\}$ where $n_G$
is the number of nodes of $G$.
Every graph $G$ has an associated $n_G \times n_G$ \emph{adjacency matrix}, also
denoted by $G$, where $G_{xy}$ is the number of arcs from $x$ to $y$.

A \emph{(simple) undirected graph} is a loopless\footnote{Note that our negative results are \emph{a fortiori} true if we consider
undirected graphs with loops. Our positive results are still valid in the same case
using the standard convention that each loop increases the degree by two.} graph $G$ such that for all $x,y \in N_G$
we have $G_{xy}=G_{yx}\leq 1$. In other words, there is at most one arc between any two nodes and if there is an 
arc from $x$ to $y$ there is also an arc in the opposite direction.
In an undirected graph, an \emph{edge} between $x$ and $y$ is a pair of arcs
$x\to y$ and $y\to x$, and it is denoted by $x\adj y$. This definition is equivalent to the more common
notion that an edge is an unordered set of nodes, but it makes it
possible to mix undirected and directed graphs: indeed, even in drawings
we will freely mix arcs and edges. For undirected graphs,
we prefer to use the word ``vertex'' instead of ``node''.

\section{Score and rank monotonicity axioms on undirected graphs}

One of the most important notions that researchers have been trying to capture
in various types of graphs is ``node centrality'':
ideally, every node (often representing an
individual) has some degree of influence or importance within the social domain
under consideration, and one expects such importance to be reflected in the
structure of the social network; centrality is a quantitative measure that
aims at revealing the importance of a node.

Formally, a \emph{centrality} (measure or index) is any function $c$ that, given a graph $G$, assigns a  
real number $c_G(x)$ to every node $x$ 	of $G$; countless notions of centrality have been proposed over time, for
different purposes and with different aims; some of them were originally defined only for a specific category of graphs. Later some of
these notions of centrality have been extended to more general classes; all centrality measures discussed in this
paper can be defined properly on all undirected graphs (even disconnected ones).
We assume from the beginning that
the centrality measures under examination are invariant by isomorphism, that is, that they depend just on the
structure of the graph, and not on a particular name chosen for each node. In particular, all nodes exchanged
by an autorphism necessarily share the same centrality score, and we will use this fact to simplify our computations.

Axioms are useful to isolate properties of different centrality measures and make it possible to compare them. One
of the oldest papers to propose this approach is~\cite{SabCIG}, which introduced score monotonicity,
and many other proposals have appeared in the
last few decades.

In this paper we will be dealing with two properties of centrality measures:

\begin{defi}[Score monotonicity]
Given an undirected graph $G$, 
a centrality $c$ is said to be \emph{score monotone on $G$} iff for every pair of non-adjacent vertices $x$ and $y$ we have that
\[
	c_G(x) < c_{G'}(x) \quad\text{and}\quad  c_G(y) < c_{G'}(y),
\]
where $G'$ is the graph obtained adding the new edge $x \adj y$ to $G$.
We say that $c$ is \emph{score monotone on undirected graphs} iff it is score monotone on all 
undirected graphs.
\end{defi}

\begin{defi}[Rank monotonicity]
Given an undirected graph $G$, 
a centrality $c$ is said to be \emph{rank monotone\footnote{We remark that in~\cite{BFVSRMUN}
rank monotonicity was defined incorrectly, using an apparently (but not effectively)
equivalent condition to stated in~\cite{BLVRMCM} and~\cite{CDKLEAA}.} on $G$} iff for every pair of non-adjacent vertices $x$ and $y$ we have that for all vertices $z\neq x,y$
\[
	 c_G(z) < c_{G}(x) \Rightarrow  c_{G'}(z) < c_{G'}(x)  \quad\text{and}\quad
	 c_G(z) < c_{G}(y)  \Rightarrow c_{G'}(z) < c_{G'}(y), 
\]
and moreover
\[
	 c_G(z) \leq c_{G}(x) \Rightarrow  c_{G'}(z)\leq c_{G'}(x)  \quad\text{and}\quad
	 c_G(z)\leq c_{G}(y)  \Rightarrow c_{G'}(z) \leq c_{G'}(y),
\]
where $G'$ is the graph obtained adding the new edge $x \adj y$ to $G$.
It is said to be \emph{strictly rank monotone on $G$} if instead 
\[
	 c_G(z) \leq c_{G}(x) \Rightarrow  c_{G'}(z) < c_{G'}(x) \quad\text{and}\quad
	 c_G(z) \leq c_{G}(y)  \Rightarrow c_{G'}(z) < c_{G'}(y) 
\]
We say that $c$ is \emph{(strictly) rank monotone on undirected graphs} iff it is (strictly) rank monotone on all 
undirected graphs.
\end{defi}

Score monotonicity tells us that in absolute terms the new edge is beneficial to $x$ and $y$. Rank monotonicity tells us
that in relative terms the new edge is not hurting them, in the sense that nodes that were (strictly) dominated by $x$ or $y$ are still (strictly) dominated. Finally,
strict rank monotonicity is a stronger property that implies, besides preservation of dominance, an improvement, as additionally all
nodes in a score tie with $x$ or $y$ will have a strictly smaller score after adding the new edge. As a sanity
check, we note that degree, the simplest centrality measure, is both score monotone and strictly rank monotone.

These three properties can be studied on the class of all undirected graphs or only on the class of connected graphs, giving rise to six possible ``degrees
of monotonicity'' that every given centrality may satisfy or not. This paper studies these different degrees of monotonicity for some of the most popular 
centrality measures, also comparing the result obtained with the corresponding properties in the directed case.

With respect to the directed case, there is an important difference: violation of the axioms may happen on one of
the nodes involved, or on both. While we never witnessed the latter situation, there is in the first case 
a distinction that we feel important enough to deserve a name: 
\begin{defi}
A violation of score monotonicity is a \emph{top violation} if the endpoint of the new edge whose scores decreases is more
important than the other. It is a \emph{bottom violation} otherwise. The same distinction applies to violations of rank monotonicity. 
\end{defi}
Top violations are somewhat sociologically natural: if a network superstar becomes friend with a nobody, it is not surprising that
the nobody increases their popularity, whereas the superstar loses a bit of charm. Bottom violations, however, are
much less natural: in the same context, the nobody sees their importance decrease, 
nurturing in a bizarre inversion of flow the superstar popularity.

As we already anticipated, and differently from the directed case, all
centrality measures we consider, except for Seeley's index (which however is trivial in this context---see Section~\ref{sec:seeley}) will turn out to be not rank monotone. 
Moreover, most centralities are not score monotone.
As a consequence, this paper
is a sequence of counterexamples (to score monotonicity and to rank monotonicity, hence \emph{a fortiori} to its strict version):
all counterexamples exhibit an undirected graph $G$ and two non-adjacent vertices $x$ and $y$ such that when
you add the edge $x \adj y$ to $G$, $x$ decreases its score, or its rank with respect to some other vertex $z$. We may call $x$ the ``losing endpoint'' (i.e., the one that is hurt by the addition of the edge).

Not all counterexamples are equally good, though. We will make an effort to have the theoretically 
strongest counterexamples we can find, and we will also look for properties that have a practical interpretation. More in
detail:
\begin{itemize}
  \item all our counterexamples are connected;
  \item all our counterexamples are parametric graphs that can be instantiated in graphs of arbitrarily large size;
  \item we always give both top and bottom violation counterexamples; thus, even knowing whether you are more or less 
  important than your new neighbor will not help in knowing if you will gain or lose from the new edge;
  \item in all our counterexamples the losing endpoint of the new edge is also \emph{demoted}, that is, 
  the number of nodes with a larger score than the losing endpoint increases after adding the new edge. 
\end{itemize}

The last point is particularly important because demotion \emph{is not implied
by the lack of rank monotonicity}:
it may be the case that $x$ used to be more important than $z$ and it becomes
less important than $z$ after the addition of the edge $x\adj y$, but still the
number of nodes that are more important than $x$ becomes smaller with the
addition of $x \adj y$. The lack of demotion might suggest a weaker notion of rank monotonicity,
in which the number of nodes whose score dominates $x$ (or $y$) decreases (such a notion is strictly
weaker as it is implied by rank monotonicity). However, this weaker notion
is not very appealing from a practical viewpoint, as it is not \emph{locally testable}---it
has no immediate consequence for the relative importance of an endpoint of the edge and another vertex.
Proving demotion implies that the counterexamples in this paper are strong enough to violate
also the weaker notion of monotonicity described above.

\section{Geometric centralities}

Since adding a new edge can only shorten existing shortest paths or create new
ones, it is immediate to show that harmonic centrality is score monotone; for the same reason,
closeness centrality is score monotone on connected graphs, whereas
counterexamples similar to those of the directed case of~\cite{BoVAC} prove 
that closeness is not score monotone in the general case.

Less intuitively, neither closeness nor harmonic centrality are rank monotone in the undirected case.
The family of counterexamples we found shows that adding an edge can shorten distances in ways
that are much more useful for some vertices not incident on the new edge than on its endpoints.  

Our counterexample for rank monotonicity of closeness and harmonic centrality
is shown in Figure~\ref{fig:closeness}.
The idea behind the graph is that the edge $0\adj1$ reduces the distance between
vertex $0$ and the vertices labeled with $4$, but does not reduce the distance between vertex $0$ and vertex $3$ (and more
importantly between vertex $0$ and the star around vertex $3$). Thus the vertices labeled with $4$ will
gain more centrality from the new edge than vertex $0$, and for appropriate values of $j$ and
$k$ we will be able to prove a violation of rank monotonicity (all vertices labeled with  $4$ share 
the same centrality). The stars of size $r$ around vertex $1$
and vertex $2$ will instead be useful by giving us some more space to play with the 
relative importance of the endpoints of the new edge, tuning the graph in Figure~\ref{fig:closeness} to be an example
of top or bottom violation.

\subsection{Closeness}

We recall that closeness of a vertex $x$ is defined as the reciprocal of its \emph{peripherality}
\[
p(x)=\sum_{y\in N_G}d(x,y),
\]
where $d(x,y)$ is the \emph{distance} (i.e., the length of a shortest path) between $x$ and $y$.

\begin{figure}
\centering
\includegraphics{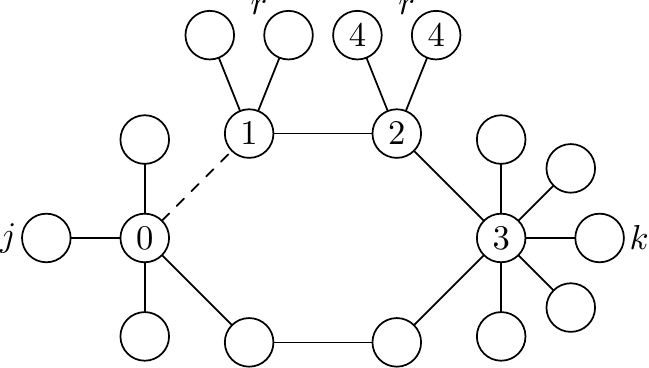}
\caption{\label{fig:closeness}A counterexample to rank monotonicity for closeness and harmonic centrality.
There is a star with $j$ leaves
around vertex $0$, a star with $k$ leaves around vertex $3$, 
a star with $r$ leaves around vertex $1$ and a star with $r$ leaves around vertex $2$.
Before adding the edge $0\adj1$, the score of vertex $0$ is larger than the score of the vertices
labeled with $4$; after, it is smaller.}
\end{figure}

We denote
for simplicity with $\pre(-)$ and $\post(-)$ the peripherality of the graph in Figure~\ref{fig:closeness} before and after adding the edge $0\adj 1$. Then,
\begin{align*}
\pre(0)&=15 + j + 4 k + 11r & \post(0)&=9 + j + 4  k + 5r \\
 \pre(1)&=15 +6j + 3  k + 3r  & \post(1)&=9 + 2j + 3  k + 3r \\
 \pre(4)&=15 + 6j + 3 k + 5r & \post(4)&=13 + 4j + 3 k + 5r. \\
\end{align*}
We are interested in finding solutions, if they exists, to the set of inequalities
\[
\pre(0) > \pre(1), \pre(0)<\pre(4),  \post(0)> \post(4),
\]
which specify that vertex $0$ violates rank monotonicity with respect to vertices labeled with $4$, and that it is less important than vertex $1$ (recall we are
manipulating the reciprocal of closeness), and
\[
\pre(0) < \pre(1), \pre(0)<\pre(4),  \post(0)> \post(4),
\]
that correspond to the analogous case in which vertex $0$ is more important than vertex $1$. There
are infinite solutions for both sets of inequalities, and in particular 
$j=5r$, $k=18r$ ($r\geq 2$), and $j=4r+4$, $k=12r+17$ ($r\geq 1$) satisfy the
first and second set, respectively.

\begin{thm}
Closeness is not rank monotone on the graphs of Figure~\ref{fig:closeness} for $r\geq 2$,
$j=5r$, and $k=18r$ (bottom violation) and for $r\geq 1$,
$j=4r+4$, and $k=12r+17$  (top violation).
\end{thm}

While the family of graphs we consider contains graphs of unbounded size, each graph has just
ten distinct peripherality scores. We can thus
compare exactly the peripherality of all vertices with that of vertex $0$ before and after
adding the new edge. It is easy to see that for the parameter
sets of the previous theorem all vertices, except the $j$ vertices labeled with $4$ and sometimes vertex $1$, maintain the same relative position
to vertex $0$ after adding the edge $0\adj 1$. Thus, in both cases vertex $0$ is demoted by at least $j-1$ 
positions.   

\subsection{Harmonic centrality}

The counterexample in Figure~\ref{fig:closeness} works also for harmonic centrality, which is
not surprising as the only difference between closeness and harmonic centrality is the usage
of a harmonic mean instead of an arithmetic mean.

Denoting this time with $\pre(-)$ and $\post(-)$ the harmonic centrality of the graph in Figure~\ref{fig:closeness} before and after adding the edge $0\adj 1$,
we have
\begin{align*}
\pre(0)&=\frac{137}{60} + j + \frac14k + \frac{11}{30}r  & \post(0)&=\frac{10}{3} + j + \frac14k + \frac56r \\
\pre(1)&=\frac{137}{60} + \frac16j + \frac13k +\frac32r  & \post(1)&=\frac{10}{3} + \frac12j + \frac13k + \frac32r\\ 
    \pre(4)&= \frac{137}{60}+ \frac16j + \frac13k + \frac56r 
  & \post(4)&=\frac{29}{12} + \frac14j + \frac13k + \frac56r . \\
\end{align*}

This time we are interested in finding solutions, if they exists, to the set of inequalities
\[
\pre(0) < \pre(1), \pre(0)>\pre(4),  \post(0)< \post(4)
\]
and
\[
\pre(0) >\pre(1), \pre(0)>\pre(4),  \post(0)< \post(4).
\]
There
are again infinite solutions for both sets of inequalities, and in particular 
$j=26r$, $k=247r$ ($r\geq 1$) and $j=26r$, $k=246r$ ($r\geq 1$), satisfy the
first and second set, respectively.

\begin{thm}
Harmonic centrality is not rank monotone on the graphs of Figure~\ref{fig:closeness} for $r\geq 1$,
$j=26r$, and $k=247r$ (bottom violation) and for $r\geq 1$,
$j=26r$, and $k=246r$  (top violation).
\end{thm}

Also in this case, for
the same parameter sets, all vertices, except the $j$ vertices labeled with $4$ and sometimes vertex $1$, maintain the same relative position
to vertex $0$ after adding the edge $0\adj 1$. Thus, vertex $0$ is demoted by at least $j-1$ 
positions.

\section{Betweenness}


\begin{figure}
\centering
\includegraphics{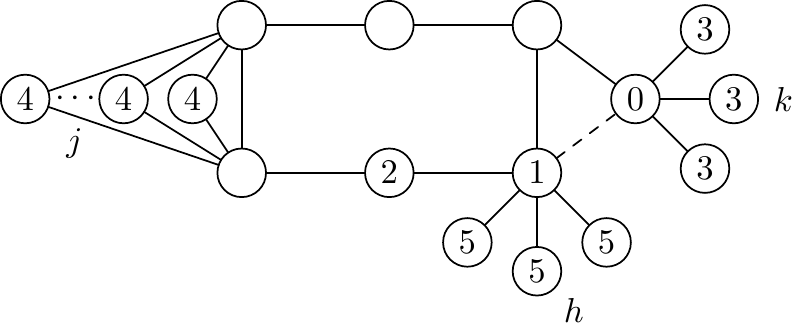}
\caption{\label{fig:betweenness}A counterexample to score and rank monotonicity for betweenness.
There is a star with $k$ leaves around vertex $0$, a star with $h$ leaves around vertex $1$, and $j$ vertices
labeled with $4$ with the same neighborhood.
Before adding the edge $0\adj 1$, the score of vertex $0$ is larger than the score of vertex $2$; after the addition, it becomes smaller.
Moreover, the score of vertex $0$ does not change when the edge is added.}
\end{figure}

Betweenness is neither score nor rank monotone on directed graphs~\cite{BLVRMCM}; the same is true in the undirected case, as
shown in the graph of Figure~\ref{fig:betweenness}.
Intuitively, the new edge puts $2$ on many shortest paths (e.g., those between any
replica of $3$ and any replica of $4$) that before needed to pass on the upper route
of the rectangle. Vertex $0$, instead, does not gain as much by the addition of the edge.

Denoting
with $\pre(-)$ and $\post(-)$ the value of betweenness before and after adding the edge $0\adj 1$, we have
\begin{align*}
\pre(0)&=\frac{k(2h+2j+k+11)}2 & 
\post(0)&=\frac{k(2h+2j+k+11)}2\\
\pre(1)&=\frac{h^2+(2j+2k+11)h+3k+7}2 & 
\post(1)&=\frac{h^2+(2j+2k+11)h+(k+1)(j+4)+4}2\\ 
\pre(2)&=\frac{(2h+2)j+3h+k+5}2 & 
\post(2)&=\frac{(2h+k+2)j+3h+2k+6}2 . \\
\end{align*}

Observe that $\pre(0)=\post(0)$, showing that score monotonicity is violated.
To prove that also rank monotonicity does not hold, we are interested in finding solutions to the set of inequalities
\[
\pre(0) < \pre(1), \pre(0)>\pre(2),  \post(0)< \post(2)
\]
and
\[
\pre(0) >\pre(1), \pre(0)>\pre(2),  \post(0)< \post(2).
\]
There
are infinite solutions for both sets of inequalities, and in particular 
$h=k$, $j=\bigl\lfloor(k^2-4k-15)/2\bigr\rfloor$, $k\geq 13$ and $k=2+h$, $j=4h$, $h\geq 12$ satisfy the
first and second set, respectively.

\begin{thm}
Betweenness is not rank monotone on the graph of Figure~\ref{fig:betweenness},
for $k=2+h$, $j=4h$, $h\geq 12$, (top violation) and 
for $h=k$, $j=\bigl\lfloor(k^2-4k-15)/2\bigr\rfloor$, $k\geq 13$ (bottom violation).\end{thm}

Also in this case we have just nine different betweenness scores, which makes it possible
to show that in both cases vertex $0$ is demoted by at least one position.
 
\section{Eigenvector centrality}

Eigenvector centrality is probably the oldest attempt at deriving a centrality
from matrix information: a first version was proposed by~\cite{LanZRWT} for
matrices representing the results of chess tournaments, and it was defined in
full generality by~\cite{BerTGA}; it was rediscovered many times since then.
One considers the adjacency matrix of the graph and computes its left or right
dominant eigenvector (in our case, the two eigenvectors coincide): the result is thus defined
modulo a scaling factor, and if the graph is (strongly) connected, the result is
unique (again, modulo the scaling factor) by the Perron--Frobenius theorem~\cite{BePNMMS}.

It is not difficult to find anecdotal examples of violation of rank (and even score, fixing a normalization) 
monotonicity in simple examples. 
\begin{figure}
\centering
\includegraphics{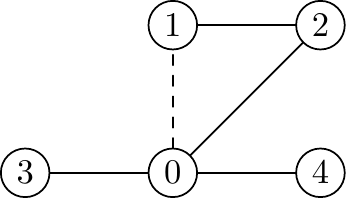}
\caption{\label{fig:ec}A counterexample to score monotonicity for eigenvector centrality. After adding the edge $0\adj1$, the score of vertex $0$ decreases:
in norm $\ell_1$, from $0.30656$ to $0.29914$; in norm $\ell_2$, from $0.65328$ to $0.63586$; and when projecting the
constant vector $\mathbf1$ onto the dominant eigenspace, from $1.39213$ to $1.35159$.}
\end{figure}
In Figure~\ref{fig:ec} we show a very simple graph that does not satisfy score monotonicity under the most obvious forms of normalization.
In particular, the score of vertex $0$ decreases after adding the edge $0\adj1$ both
in norm $\ell_1$ and norm $\ell_2$, and even when projecting the constant vector $\mathbf1$ onto the dominant eigenspace, which is an alternative way of circumventing the
scaling factor~\cite{VigSR}. The intuition 
is that initially vertex $0$ has a high score because of its largest degree (three). However, once we close the triangle
we create a cycle that absorbs a large amount of rank, effectively decreasing the score of vertex $0$. 
%
%

\begin{figure}
\centering
\includegraphics{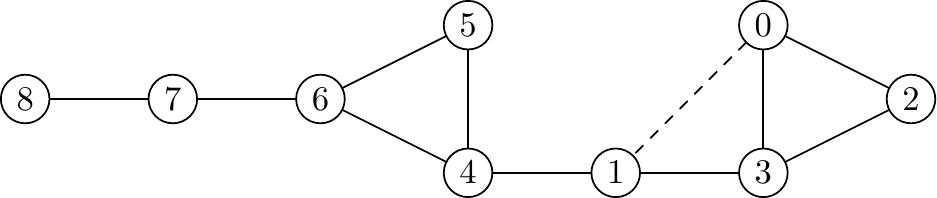}
\caption{\label{fig:ec2}A counterexample to rank monotonicity for eigenvector centrality.
Before adding the edge $0\adj 1$, the score of vertex $1$ is larger than the score of vertex $3$; after, it is smaller.}
\end{figure}

A similar counterexample, shown in Figure~\ref{fig:ec2}, proves that eigenvector centrality does not satisfy rank monotonicity.
Before adding the edge $0\adj 1$, the score of vertex $1$ used to be larger than the score of vertex  $3$; the converse is true after the addition
of the edge. This counterexample,
however, is not very satisfactory as vertex $1$ is not demoted---in fact, the opposite happens; on the other hand,
the set of vertices that dominate it changes completely with the addition of the new edge, showing that eigenvector
centrality can undergo turbulent modifications upon a simple perturbation.

We are now going to prove that eigenvector centrality does not satisfy
rank monotonicity on a class of graphs of arbitrarily large size in which we will also experience demotion.
Proving analytical results will require combining a few techniques from spectral
graph theory and analysis, as we would otherwise not be able to perform exact
computations, as in the previous cases.

\section{Interlude: graph fibrations}
\label{sec:fib}

Proving analytical results about graphs of arbitrary size requires in principle manipulating matrices
of arbitrary size, and obtaining closed-form expressions for eigenvalues and eigenvectors of such
matrices would be difficult, if not impossible. We thus turn to ideas going back to the results
obtained in the '60s in the context of the theory of \emph{graph divisors}~\cite{SacUTFCPG}, restating
them in the more recent language of \emph{graph fibrations}~\cite{BoVGF}.

A \emph{(graph) morphism} $\phi:G\to H$ is given by a pair of functions
$f_N:N_G\to N_H$ and $f_A:A_G\to A_H$ commuting with the source and
target maps, that is, $s_H(f_A(a))=f_N(s_G(a))$ and $t_H(f_A(a))=f_N(t_G(a))$ for all
$a \in A_G$. In other
words, a morphism maps nodes to nodes and arcs to arcs in such a way to
preserve the incidence relation.  
The definition of morphism we give is the obvious extension to the case of multigraphs of the standard notion the
reader may have met elsewhere.

\begin{defi}
\label{def:fibration}
A \emph{fibration}~\cite{BoVGF,GroTDTEGAI} between the graphs $G$ and $B$ is a morphism $\phi: G\to B$ such
that for each arc $a\in A_B$ and each node $x\in N_G$ satisfying
$\phi_N(x)=t_B(a)$ there is a unique arc $\lift ax\in A_G$ (called the \emph{lifting of
$a$ at $x$}) such that $\phi_A(\lift ax)=a$ and $t_G(\lift ax)=x$.
\end{defi}

If $\phi:G\to B$ is a fibration, $G$
is called the \emph{total graph} and $B$ the \emph{base} of $\phi$. 
We shall also say that $G$ is \emph{fibered (over $B$)}. The \emph{fiber over a
node $x\in N_B$} is the set of nodes of $G$ that are mapped to $x$.

A verbal restatement of the definition of fibration
is that each arc of the base lifts uniquely to each node in the fiber of its target;
moreover, we remark that Definition~\ref{def:fibration} is just an elementary restatement 
of Grothendieck's notion of fibration between categories applied to the
free categories generated by $G$ and $B$.

In Figure~\ref{fig:exfib}, we show two graph morphisms; the morphisms are
implicitly described by the colors on the nodes and in the only possible way on the arcs. The morphism displayed on the
left is not a fibration, because the loop
on the base has no counterimage ending at the lower gray node, and
moreover the other arc has two counterimages with the same target. The
morphism displayed on the
right, on the contrary, is a fibration. Observe that loops are not necessarily
lifted to loops.

\begin{figure}[htbp]
  \begin{center}
	\includegraphics{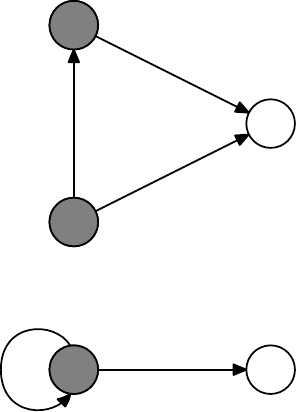}\qquad\qquad\qquad\qquad\includegraphics{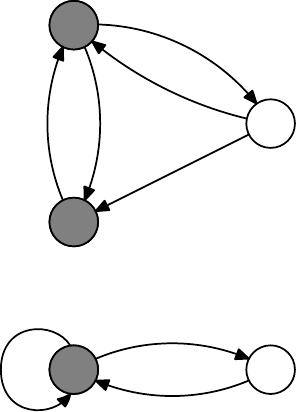}
  \end{center}
  \caption{\label{fig:exfib}On the left, an example of graph morphism that is
not a fibration; on the right, a fibration. Colors on the nodes are used to
implicitly specify the morphisms (arcs are mapped in the only possible way).}
\end{figure}

\begin{defi}
If $\phi:G\to B$ is a fibration, given a (row) vector $\bm u$ of size $n_B$, define its \emph{lifting along $\phi$} as
the vector $\bm u^\phi$  of size $n_G$ given by
\[
	\left(\bm u^\phi\right)_i=u_{\phi(i)}.
\]
\end{defi}
Otherwise said, $\bm u^\phi$ is the vector obtained by copying $\bm u$ along the fibers of $\phi$.

\begin{thm}[\cite{SacUTFCPG}]
\label{th:sachs}
If $\phi:G\to B$ is a fibration surjective on the nodes, given a (row) vector $\bm u$ of size $n_B$ we have
\[
\bm u^\phi G = (\bm u B)^\phi.
\]
\end{thm}

In other words, one can lift and multiply by $G$, or equivalently multiply by $B$ and then lift: the base
$B$ ``resumes'' the graph $G$ well enough that the multiplication of fiberwise constant vectors by $G$
can be carried on (usually smaller) $B$. The proof of Theorem~\ref{th:sachs} is in fact immediate
once one realizes that Definition~\ref{def:fibration} implies that $\phi$ induces \emph{a local isomorphism} between
the in-neighborhood of a node $x$ of $G$ and the in-neighborhood of $\phi_N(x)$~\cite{BoVGF}.

Theorem~\ref{th:sachs} has the important consequence that every left eigenvector $\bm e$ of $B$ can be lifted to
a left eigenvector $\bm e^\phi$ of $G$, so every eigenvalue of $B$ is an eigenvalue of $G$, and
thus the characteristic polynomial of $B$ divides that of $G$ (hence the name \emph{graph divisor}).
In our case, by the Perron--Frobenius theorem~\cite{BePNMMS},
if $B$ is strongly connected the dominant eigenvector of $B$
is strictly positive, so its lifting is strictly positive, and thus (applying again the Perron--Frobenius theorem) it is the dominant eigenvector of $G$; moreover,
$G$ and $B$ share the same dominant eigenvalue (and thus spectral radius). 

\section{Back to eigenvector centrality}
\label{sec:eigen}

We now get back to eigenvector centrality: Figure~\ref{fig:eig} shows a family of total graphs $G_k$ depending
on an integer parameter $k$, and an associated family of bases $B_k$, with fibrations defined on the nodes following the node labels,
and on the arcs in the only possible way. We will show that when the edge $0\adj 1$ is added to the graphs (obtaining new graphs $G_k'$ and
$B_k'$), all vertices
labeled with $4$, which used to have a smaller score than vertex $1$ in $G_k$, will become more important than vertex $1$ in $G_k'$. 

The intuitive idea behind the graphs $G_k$ is that the new edge makes the vertices labeled with $4$ much closer
to vertex $1$, a high-degree vertex; at the same time, the new edge doubles the number of paths from the vertices
labeled with $6$ to the vertices labeled with $4$. The advantage for vertex $1$ is to get much closer to the
vertices labeled with $4$, but those have a much smaller degree. All in all, the new edge will turn out to be much more advantagous
for the vertices labeled with $4$ than for vertex $1$.

The fundamental property of our counterexample is that albeit $G_k$ is a simple undirected graph
with $k^2-k-6$ vertices, $B_k$ is a general directed multigraph with seven nodes, independently of $k$,
so its adjacency matrix, shown
in Figure~\ref{fig:eig}, is a fixed-sized matrix containing a parameter $k$ due
to the variable number of arcs. Thus, fibrations make it possible to move our
proof from matrices of arbitrary size to a parametric matrix of fixed size.

\begin{figure}
\centering
\begin{tabular}{cc}
\raisebox{3cm}{$G_k$\qquad}&\includegraphics{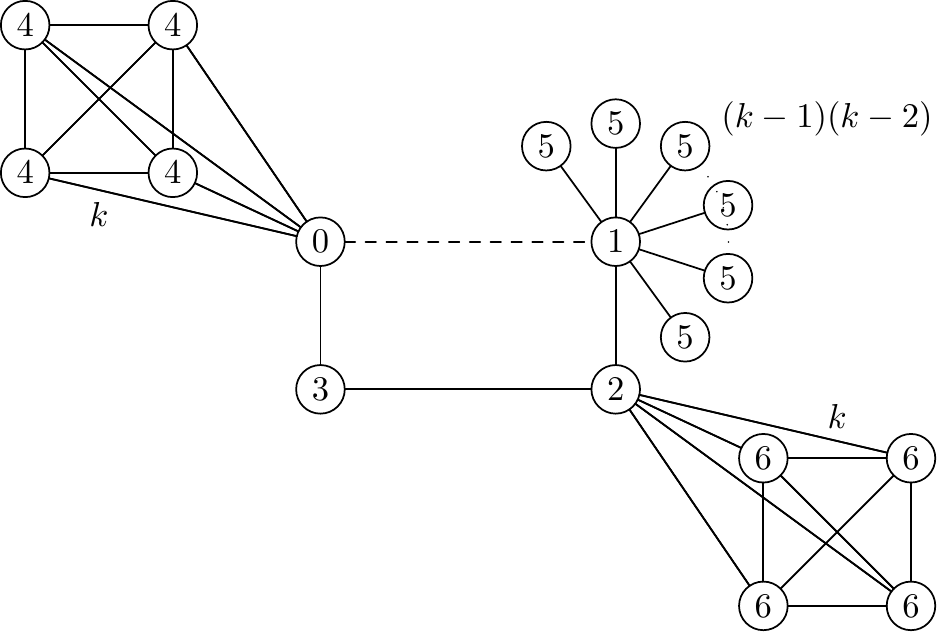}\\
\raisebox{2cm}{$B_k$\qquad}&\includegraphics{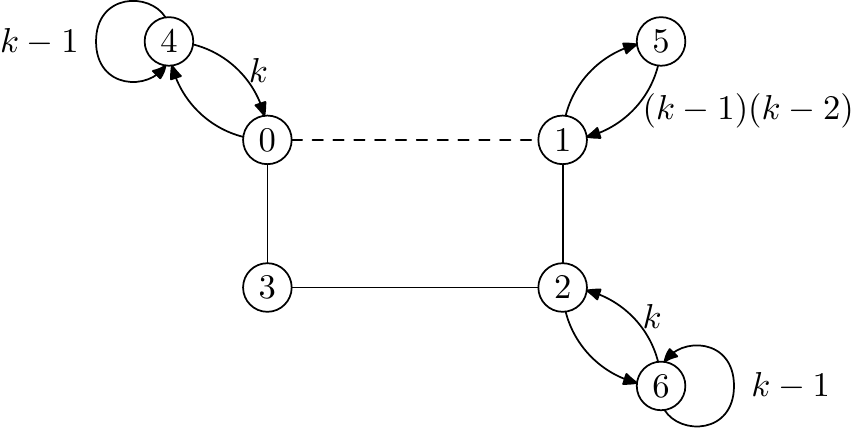}
\end{tabular}

\[
\setlength\arraycolsep{2ex}
\renewcommand\arraystretch{1.5}
B_k=\left(\begin{matrix}
    0 & \fcolorbox{gray}{gray}{0} & 0 & 1 & 1 & 0 & 0\\
    \fcolorbox{gray}{gray}{0} & 0 & 1 & 0 & 0 & 1 & 0\\
    0 & 1 & 0 & 1 & 0 & 0 & 1\\
    1 & 0 & 1 & 0 & 0 & 0 & 0\\
    k & 0 & 0 & 0 & k-1 & 0 & 0\\
    0 & (k-1)(k-2) & 0 & 0 & 0 & 0 & 0\\
    0 & 0 & k & 0 & 0 & 0 & k-1
\end{matrix}\right)
\]

\caption{\label{fig:eig}The parametric counterexample graph for eigenvector centrality:
when adding the edge $0\adj1$ vertex $1$ violates rank monotonicity (top).
The $k$ vertices labeled with $4$ 
form a $(k+1)$-clique with vertex $0$, and the $k$ vertices labeled with $6$ form a $(k+1)$-clique
with vertices $2$; finally, there is a star with $(k-1)(k-2)$ leaves around vertex $1$. 
Arc
labels represent multiplicity. The matrix displayed is the adjacency matrix of $B_k$, with the grayed entries
to be set to $1$ when $0\adj1$ is added to the graph. Table~\ref{tab:eig} shows a set of values for
the size of the cliques and the size of the star causing vertex $1$ to be less important than vertex $0$.}
\end{figure}

\subsection{Sturm polynomials}

There is no way to compute exactly the eigenvalues and eigenvectors of $B_k$. However,
we will be able to control their behavior using \emph{Sturm polynomials}~\cite{RaSATP}, a standard, powerful
technique to analyze and locate real roots of polynomials.

\begin{defi}
If $p(x)$ is a polynomial with real coefficients and $p'(x)$ its derivative, the \emph{Sturm sequence} of polynomials
associated with $p(x)$ is defined by
\begin{align*}
S_0(x) &= p(x)\\
S_1(x) &= p'(x)\\
S_{i+1}(x) &= - S_{i}(x) \bmod S_{i-1}(x)\qquad \text{for $i\geq 1$,} 
\end{align*}
where $S_{i}(x) \bmod S_{i-1}(x)$ is the remainder of the Euclidean division of $S_i(x)$ by $S_{i-1}(x)$.
The sequence stops when $S_{i+1}(x)$ becomes zero, and it is long at most as the degree of $p(x)$.
\end{defi}

Given a real number $a$, the number of \emph{sign variations} $V(a)$ of a Sturm sequence is the number of
sign changes, ignoring zeros, of the sequence $S_0(a)$,~$S_1(a)$,~$S_2(a)$, $\dots\,$. Finally,
if $p(x)$ is \emph{squarefree} (i.e.,
it is not divisible by the square of a noncostant polynomial),
the number of distinct roots of $p(x)$ in the interval $(a\..b]$ is $V(a)-V(b)$; all polynomials we will study will be squarefree.

\subsection{Bounding the dominant eigenvalue}

We now discuss how to bound the dominant eigenvalue $\rho_k$ of $B_k$ (and thus $G_k$); the same results
hold for the dominant eigenvalue $\rho'_k>\rho_k$ of $B'_k$ (and thus $G'_k$). 
The approach we describe will be used
throughout the rest of the paper.
 
Consider the characteristic polynomial of $B_k$ 
\[
p_k(\lambda) = \det(1- \lambda B_k).
\]
We can compute its Sturm polynomials and evaluate them at the points
$k+\frac1{k^2}$ and $k+\frac3{4k}$.
This evaluation leaves us with a pair of rational functions in $k$ for each
Sturm polynomial in the sequence, and such functions have a defined sign
for $k\to\infty$ that depends on the sign of the ratio of the leading
coefficients of their numerator and denominator:
in other words, for large enough $k$ we can count the number of zeroes of
$p_k(\lambda)$ in the interval
$(k+\frac1{k^2}\..k+\frac3{4k}]$, and indeed $p_k(\lambda)$
has exactly one zero in that interval for $k\geq 24$.

If we apply the same technique to the interval $\left(k+\frac3{4k}\..2k\right]$, we find no zeroes. Since $2k$ is
an upper bound for the dominant eigenvalue of both matrices (as it is larger than the geometric mean
of indegree and outdegree of all vertices~\cite{KwaSRDG}), we conclude that the spectral radius $\rho_k$ of $B_k$ 
lies in $\left(k+\frac1{k^2}\..k+\frac3{4k}\right]$. 

\subsection{Bounding the dominant eigenvector}

Armed with this knowledge, we approach the study of the dominant eigenvectors of $B_k$ and $B'_k$.
There is no way to compute them exactly: thus, we resort to the study of
$\mathbf 1(1 -\alpha B_k\bigr)^{-1}$,
because the dominant eigenvector $\bm e$ of $B_k$ and $\bm e'$ of $B'_k$ can be expressed as~\cite{VigSR}
\begin{align}
\label{eq:limit}
\bm e &= \lim_{\alpha\to1/\rho_k}\bigl(1-\alpha\rho_k\bigr)\mathbf1\bigl(1 -\alpha B_k\bigr)^{-1}.\\
\bm e' &= \lim_{\alpha\to1/\rho'_k}\bigl(1-\alpha\rho'_k\bigr)\mathbf1\bigl(1 -\alpha B'_k\bigr)^{-1}.
\end{align}
In fact,
$(1 -\alpha B_k\bigr)^{-1}$ is a slightly different
way (up to a constant factor) to define the \emph{resolvent} of $B_k$~\cite{DuSLO1}, but the formulation we use here will make it easier to apply the results
we will develop in the sections on Katz's index and PageRank.

While we have no way to compute exactly the eigenvectors of $B_k$, we can 
compute symbolically $\mathbf1\bigl(1 -\alpha B_k\bigr)^{-1}$, thus obtaining for each node of $B_k$
a rational function in $\alpha$ whose coefficients are polynomials in $k$, and do the same for $B'_k$.

We will be interested in comparing eigenvector centralities, that is, in proving statements (for nodes $x$ and $y$ of $B_k$) of the form
\[
\frac{e_x}{e_y}=\lim_{\alpha\to1/\rho_k}\frac{\left[\bigl(1-\alpha\rho_k\bigr)\mathbf1\bigl(1 -\alpha B_k\bigr)^{-1}\right]_x}{\left[ \bigl(1-\alpha\rho_k\bigr)\mathbf1\bigl(1 -\alpha B_k\bigr)^{-1}\right]_y}>1.
\]
However,   
\[
\frac{e_x}{e_y}=\lim_{\alpha\to1/\rho_k}\frac{\left[\mathbf1\bigl(1 -\alpha B_k\bigr)^{-1}\right]_x}{\left[ \mathbf1\bigl(1 -\alpha B_k\bigr)^{-1}\right]_y}
=\lim_{\alpha\to1/\rho_k}\frac{\left[\mathbf1\cdot\adjugate{1-\alpha B_k}\right]_x}{\left[\mathbf1\cdot\adjugate{1-\alpha B_k}\right]_y}
 =\frac{\left[\mathbf1\cdot\adjugate{1- B_k/\rho_k}\right]_x}{\left[\mathbf1\cdot\adjugate{1- B_k/\rho_k}\right]_y},
\]
where we used the fact that the inverse is the \emph{adjugate matrix}~\cite{GanTM} divided by the determinant
\[
	\adjugate{1-\alpha B_k}=\bigl(1 -\alpha B_k\bigr)^{-1}\cdot \det(1 -\alpha B_k).
\]
The final substitution can be performed safely because the column-sums of the adjugate
must be nonzero in a neighborhood of $\rho_k$, or the limits~(\ref{eq:limit}) would not be finite and positive.
The advantage is that the entries of $\adjugate{1-\alpha B_k}$ are just polynomials. The same considerations hold for $B'_k$.

We thus define, for every node $x$,
\begin{align*}
 \pre_\alpha(x)&= \left[\mathbf1\cdot\adjugate{1 -\alpha B_k}\right]_x\\
 \post_\alpha(x)&=  \left[\mathbf1\cdot\adjugate{1 -\alpha B'_k}\right]_x.\\
\end{align*}
For example,
\begin{multline*}
\pre_\alpha(0) = (-2k^3 + 7k^2 - 7k + 2)\alpha^6 + (2k^2 - 7k + 5)\alpha^5 + (2k^3 - 6k^2 + 6k)\alpha^4\\ + (k^3 - 5k^2 + 9k - 7)\alpha^3 + (-k^2 + k - 3)\alpha^2 + (-k + 2)\alpha + 1.
\end{multline*}
Note that in the adjacency matrix of $B_k$ just three rows contain $k$: as a consequence, the degree in $k$
of the coefficients of the polynomials in $\alpha$ is at most three.

Since $k+\frac3{4k}>\rho_k$, we start by showing that 
\[
\pre_{1/\left(k+\frac3{4k}\right)}(1)>\pre_{1/\left(k+\frac3{4k}\right)}(4)
\]
and once again, since we are dealing with rational functions in $k$, for enough large $k$ the difference 
\[
\pre_{1/\left(k+\frac3{4k}\right)}(1)-\pre_{1/\left(k+\frac3{4k}\right)}(4)
\]
has a constant sign: in particular, for $k\geq 53$ it is positive. The same analysis, however, shows that 
\[
\post_{1/\left(k+\frac3{4k}\right)}(1) < \post_{1/\left(k+\frac3{4k}\right)}(4)
\] 
when $k\geq 3$.

We are now going to extend our inequalities to a range comprising $1/\rho_k$. If we consider the Sturm polynomials (in $\alpha$)
of \[
\pre_\alpha(1)- \pre_\alpha(4),\] we find no zero between $\alpha = 1/\left(k+\frac3{4k}\right) <1/ \rho_k$
and $\alpha = 1/\left(k+\frac1{k^2}\right) >1/ \rho_k$ for $k\geq 53$. Hence, for $1/\left(k+\frac3{4k}\right)< \alpha\leq 1/\left(k+\frac1{k^2}\right)$
\[
\pre_\alpha(1)>
\pre_\alpha(4),
\] 
so, in particular,
\[
\pre_{1/\rho_k}(1)>\pre_{1/\rho_k}(4),
\] 
showing that the eigenvector centrality of node $1$ is larger than that of node $4$ for $k\geq 53$.
A similar analysis for $\post$ shows that 
\[
\post_{1/\rho'_k}(1)<\post_{1/\rho'_k}(4)
\] 
for $k\geq 1$.
Thus, in the graph $G_k$ the addition of the edge $0\adj1$ causes
vertex $1$ to violate rank monotonicity. Further analysis of the same kind on the remaining nodes show
that only the vertices labeled with $4$ change their importance relatively to vertex $1$,
which implies that vertex $1$ is demoted by $k$ positions.
Finally, studying the polynomial $\pre_\alpha(1)-\pre_\alpha(0)$ it is easy to see that in our example
vertex $1$ is more important than vertex $0$ for $k\geq 54$.

While all the previous discussions are valid for $k\geq 54$, numerical computations show that the result indeed extends to all $k\geq 7$.
Hence:

\begin{thm}
Eigenvector centrality is not rank monotone (top violation) on the graphs $G_k$ of Figure~\ref{fig:eig} for $k\geq 7$.
\end{thm}

By gaging accurately the size of the star around $1$ it is possible to
find also bottom violations of rank monotonicity. We have
tabulated the first few values of $k$ for which there is a suitable star, and we
show them in Table~\ref{tab:eig}: we conjecture that there is a function of $k$ of
order $\Theta(k^2)$ which gives a correct real value for $s$, and examples emerge when
such value is very close to an integer. 

\begin{table}
\centering
\begin{tabular}{ll|ll|ll|ll|ll|ll}
\multicolumn{1}{c}{$k$}&\multicolumn{1}{c|}{$s$}&\multicolumn{1}{c}{$k$}&\multicolumn{1}{c|}{$s$}&\multicolumn{1}{c}{$k$}&\multicolumn{1}{c|}{$s$}&\multicolumn{1}{c}{$k$}&\multicolumn{1}{c|}{$s$}&\multicolumn{1}{c}{$k$}&\multicolumn{1}{c|}{$s$}&\multicolumn{1}{c}{$k$}&\multicolumn{1}{c}{$s$}\\
\hline
8 & 40 &   17 & 217 & 30 & 733 &    40 & 1344  &  57 & 2815 & 68 & 4059 \\
9 & 53 &   18 & 246 & 31 & 786 &    43 & 1564  &  59 & 3024 & 69 & 4184 \\
10 & 67 &  19 & 276 & 32 & 840 &    44 & 1641  &  61 & 3241 & 70 & 4310 \\
11 & 83 &  24 & 456 & 34 & 955 &    45 & 1720  &  62 & 3352 & 72 & 4569 \\
12 & 101 & 26 & 541 & 35 & 1015&    48 & 1968  &  63 & 3465 & 73 & 4701 \\
14 & 142 & 27 & 586 & 36 & 1077 &   50 & 2143  &  64 & 3580 & 74 & 4835 \\
15 & 165 & 28 & 633 & 37 & 1141 &   51 & 2233  &  65 & 3697 & 75 & 4971 \\
16 & 190 & 29 & 682 & 38 & 1207 &   56 & 2713  &  66 & 3816 & 76 & 5109 \\
\end{tabular}
\vspace*{2mm}
\caption{\label{tab:eig}Pairs of values providing bottom violations of
rank monotonicity for eigenvector centrality: $k$ is the same as in Figure~\ref{fig:eig},
and $s$ is the size of the star around $1$ (in Figure~\ref{fig:eig}, $s=(k-1)(k-2)$).}
\end{table}

\section{Seeley's index}
\label{sec:seeley}

A natural variant of eigenvector centrality is Seeley's index~\cite{SeeNRI}, the steady state of the  (uniform)
random walk on the graph (for more details, see~\cite{BoVAC}). The situation here is quite different: it is a
well-known fact that if the graph is connected the steady-state probability of
vertex $x$ is simply $d(x)/2m$, where $d(x)$ is the degree of $x$---essentially, the centrality of a vertex 
is just its $\ell_1$-normalized degree. 
As a consequence:

\begin{thm}
\label{thm:seeleyrank}
Seeley's index is strictly rank monotone on undirected graphs.
\end{thm}

The situation is almost the same for score monotonicity if we assume $\ell_1$-normalization:
\begin{thm}
\label{thm:seeleyscore}
Seeley's index ($\ell_1$-normalized degree) is score monotone on undirected graphs, except in the case of a graph formed by a star graph and one or more additional isolated vertices.
\end{thm} 
\begin{proof}
When we add an edge between $x$ and $y$ in a graph with $m$ edges, the score of $x$ changes from $d(x)/2m$ to $(d(x)+1)/(2m+2)$. If we require
\[
\frac{d(x)+1}{2m+2}> \frac{d(x)}{2m}
\]
we obtain $d(x)< m$.  Since obviously $d(x)\leq m$, the condition is always true except when $d(x)=m$, which corresponds to
the case of a disconnected graph formed by a star graph and by additional isolated vertices.
Indeed, in that case adding an edge between an isolated vertex and the center of the star will not change the score
of the center.
\end{proof}

\section{Interlude: graph fibrations and damped spectral rankings}

The key observation used to build the counterexample for eigenvector centrality was Theorem~\ref{th:sachs},
stating that lifting of vectors commutes with matrix multiplication. 

The theorem is true also for \emph{weighted} graphs, as long as the fibration preserves weights
and adjacency matrices are defined by adding the weights of all arcs between two nodes. 
An interesting consequence of this fact is the following:
\begin{thm}
\label{th:sachsex}[\cite{BLSGFGIP}]
Let $G$ and $B$ be weighted graphs, and $\phi:G\to B$ be a surjective weight-preserving fibration; then,
given a (row) vector $\bm v$ of size $n_B$ we have
\[
\bm v^\phi (1-\alpha G)^{-1} = \bigl(\bm v (1-\alpha B)^{-1}\bigr)^\phi.
\]
\end{thm}
The proof is simple:
\begin{multline*}
\bm v^\phi (1-\alpha G)^{-1} = \bm v^\phi \sum_{i=0}^\infty (\alpha G)^i =
\sum_{i=0}^\infty\bm v^\phi  (\alpha G)^i \\= \sum_{i=0}^\infty\bigl(\bm v  (\alpha B)^i\bigr)^\phi =
\Bigl(\bm v\sum_{i=0}^\infty(\alpha B)^i\Bigr)^\phi = \bigl(\bm v (1-\alpha B)^{-1}\bigr)^\phi.  
\end{multline*}

Theorem~\ref{th:sachsex} makes it possible to apply the techniques we used for eigenvector centrality to general \emph{damped spectral rankings}, as defined 
in~\cite{VigSR}, of which both Katz's index and PageRank are special instances. Both centralities can be defined, up to a constant multiplying factor, as
\[
\bm v (1-\alpha M)^{-1}
\]
for suitable preference vector $\bm v$ and for a matrix $M$ derived from the adjacency matrix of the graph.

\section{Katz's index}
\label{sec:katz}

Recall that Katz's index~\cite{KatNSIDSA} is defined as
\[
\mathbf1\sum_{i=0}^\infty \alpha^i G^i = \mathbf1(1-\alpha G)^{-1},
\]
where $0\leq\alpha<1/\rho(G)$ (here, $\rho(G)$ is the spectral radius of $G$). 
It is trivially score monotone, but we will prove that it is not rank monotone.

First of all, we note that if $\alpha$ is small enough Katz's index will be strictly rank monotone:
\begin{thm}
\label{th:katzzero}
Let $G$ be a graph and $\rho$ its spectral radius.
Then there is an $\bar\alpha<1/\rho$ such that for $\alpha\leq\bar\alpha$  
Katz's index is strictly rank monotone on $G$.
\end{thm}
\begin{proof}
We remark that
\[
\mathbf1\sum_{i=0}^\infty \alpha^i G^i = \mathbf1 + \alpha\mathbf1G + \alpha^2 \mathbf1 G^2\sum_{i=0}^\infty \alpha^i G^i.  
\] 
The relative node importance in $\mathbf1 + \alpha\mathbf1G$ is exactly that defined by degree, and score differences are $O(\alpha)$ for $\alpha\to0$. However,  
\[
\Bigl\|\alpha^2\mathbf1 G^2 \sum_{i=0}^\infty \alpha^iG^i\Bigr\|_\infty \leq \alpha^2\bigl\|G^2\bigr\|_\infty\Bigl\|\sum_{i=0}^\infty \alpha^i G^i\Bigr\|_\infty,  
\]
and given any $0<\alpha'<1/\rho_k$ for $\alpha\leq\alpha'$ the last expression is $O\bigl(\alpha^2\bigr)$ for $\alpha\to 0$. Thus, there is an $\alpha_G$ such that,
for $\alpha\leq\alpha_G$, the relative importance of a node of $G$ is that defined by its degree. If we minimize over all such $\alpha$'s for all
graphs obtained by adding an edge to $G$, we obtain the value $\bar\alpha$ of the statement.
\end{proof}

On the other hand, we are now going to provide an example on which 
rank monotonicity is not satisfied when we go sufficiently close to $1/\rho$. 
We can use the same counterexample as for eigenvector centrality (Figure~\ref{fig:eig}): 
in view of Theorem~\ref{th:sachsex}
the analysis performed
in Section~\ref{sec:eigen} already shows that Katz's index is not rank monotone on $G_k$ for sufficiently large $k$ and for all 
\[
	\alpha \in \biggl[\frac1{k+\frac3{4k}}\..\frac1{\rho_k'}\biggr).
\]
In other words,

\begin{thm}
\label{thm:katz}
Let $\rho_k'$ be the spectral radius of the graph $G_k'$ in Figure~\ref{fig:eig}.
For $k\geq 54$, there exists some $\nu_k<\frac1{k+\frac3{4k}}$ such that  
Katz's index is not rank monotone (top violation) on $G_k$ for all $\alpha\in \left(\nu_k\..\frac1{\rho_k'}\right)$.
\end{thm}
Note that the theorem above claims that the violation happens in a left neighborhood of the upper bound of $\alpha$; moreover,
on the left we can get as close as desired to $0$ given a suitable $k$. This is the best possible scenario, in view of Theorem~\ref{th:katzzero}.
Also our considerations about demotion in Section~\ref{sec:eigen} transfer immediately to the present setting.

Further analysis by Sturm polynomials in the interval $\left(\frac1{k+\frac1k}\..\frac1{k+\frac3{4k}}\right]$ shows the following:
\begin{itemize}
  \item the relative importance of node $1$ and node $4$ in $\pre_\alpha(-)$ flips (node $4$ is more important than node $1$ at the beginning of the interval 
and then becomes less important, after some value of $\alpha$, say $\alpha'$); 
  \item the relative importance of node $0$ and node $1$ in $\pre_\alpha(-)$ flips (node $0$ is more important than node $1$ at the beginning of the interval
and then becomes less important, after some value of $\alpha$, say $\alpha''$);
  \item $\pre_\alpha(1)-\pre_\alpha(4)$ always dominates $\pre_\alpha(1)-\pre_\alpha(0)$.
\end{itemize}
The latter observation implies $\alpha'<\alpha''=\nu_k$, in the notation of Theorem~\ref{thm:katz}; since the relative importance of node $1$ and node $4$ remains the same in $\post_\alpha(-)$ 
(node $4$ is always more important than node $1$ in the interval after the addition of the edge), in the interval $\bigl(\alpha'\.. \nu_k\bigr)$ 
we can observe a bottom violation of rank monotonicity.

Moreover, the interval $(\alpha'\..\nu_k)$ gets closer to the upper bound $\frac1{\rho'_k}>\nu_k$ as $k$ gets larger:
\[
\frac{\frac1{\rho'_k} - \alpha'}{\frac1{\rho'_k}}\leq  \frac{\frac1{\rho'_k} - \frac1{k+\frac1k}}{\frac1{\rho'_k}}  = 1-\frac{\rho'_k}{k+\frac1k}\leq1-\frac{k+\frac1{k^2}}{k+\frac1k}\to 0 \quad\text{for $k\to\infty$}.
\]

\begin{thm}
\label{thm:katz2}
For every $k$, there is an interval of values of $\alpha$ contained in $\left(\frac1{k+\frac1k}\..\frac1{k+\frac3{4k}}\right]$
in which Katz's index is not rank monotone (bottom violation). The interval gets arbitrarily close to $\frac1{\rho'_k}$ as $k\to \infty$.
\end{thm}

As a final consideration, there is another range of validity of Theorem~\ref{thm:katz}: if we further analyze with Sturm polynomials the relative
importance of node $1$ and node $4$ in the interval $\left(\frac1{k+\frac2k}\..\frac1{k+\frac3{4k}}\right]$, we find two sign changes
in $\pre_\alpha(1)- \pre_\alpha(4)$, two sign changes
in $\pre_\alpha(1)- \pre_\alpha(0)$  and zero sign changes in $\post_\alpha(1)- \post_\alpha(4)$: thus, there is an interval
comprising $\frac1{k+\frac2k}$ in which the violation of rank monotonicity happens again. Also in this interval 
$\pre_\alpha(1)-\pre_\alpha(4)$ always dominates $\pre_\alpha(1)-\pre_\alpha(0)$, hence, we have both top violations and bottom violations; it
is also immediate to show demotion.
Figure~\ref{fig:katz} resumes graphically the results proved in this section.

\begin{figure}
\centering
\includegraphics{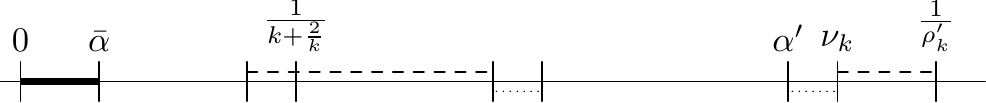}
\caption{\label{fig:katz}A graphical resume of the results about Katz's index proved in Section~\ref{sec:katz}. 
Dashed lines (dotted lines, resp.) represent intervals
of values of $\alpha$ in which we proved top (bottom, resp.) violations of rank monotonicity (Theorems~\ref{thm:katz} and~\ref{thm:katz2}).
The thick interval represents a region where rank monotonicity is guaranteed (Theorem~\ref{th:katzzero}).}
\end{figure}

\section{PageRank}

PageRank~\cite{PBMPCR} can be defined as
\[
(1-\alpha)\bm v\sum_{i=0}^\infty \alpha^i\bar G^i = (1-\alpha) \bm v(1-\alpha \bar G)^{-1},
\]
where $\alpha\in[0\..1)$ is the damping factor, 
$\bm v$ is a non-negative preference vector with unit $\ell_1$-norm, 
and $\bar G$ is the row-normalized version\footnote{Here we are assuming that $G$ has no \emph{dangling nodes} (i.e., nodes with outdegree $0$).
If dangling nodes are present, you can still use this definition (null rows are left untouched in $\bar G$), but then to obtain PageRank you 
need to normalize the resulting vector~\cite{BSVPFD,DCGRFPCSLS}. So all our discussion can also be applied
to graphs with dangling nodes, up to $\ell_1$-normalization.} of $G$; that is,
$\bar G$ is just the (adjacency matrix of the) weighted version of $G$ defined by letting $w(a)=1/\sum_{x\in N_G}G_{s_G(a)\, x}$.
Hence, if you have a weighted graph $B$,
a weight-preserving fibration $\phi: \bar G \to B$ that is surjective on the nodes, and a vector $\bm u$ of size $n_B$ such that 
$\bm u^\phi$ has unit $\ell_1$-norm, you can deduce from Theorem~\ref{th:sachsex} that
\begin{equation}
\label{eqn:resumepr}
 (1-\alpha) \bm u^\phi (1-\alpha \bar G)^{-1}=\left(\bm (1-\alpha) \bm u (1-\alpha B)^{-1}\right)^\phi.
\end{equation}

On the left-hand side you have the actual PageRank of $G$ for a preference vector that is fiberwise constant;
on the right-hand side you have a damped spectral ranking of $B$.
Note that $B$ is not necessarily row-stochastic, and $\bm u$ has not unit $\ell_1$-norm, 
so technically the right-hand side of the equation in Theorem~\ref{th:sachsex} is
not PageRank anymore.

\smallskip
We first observe that
\begin{thm}
\label{th:prlim}
Given an undirected graph $G$, there is a value of $\alpha$ for which PageRank is strictly rank monotone on $G$.
The same is true for score monotonicity, except when $G$ is formed by a star graph and one or more additional isolated vertices.
\end{thm}
\begin{proof}
We know that for $\alpha\to 1$, PageRank tends to Seeley's index~\cite{BSVPFDF}. Since Seeley's index is strictly rank monotone (Theorem~\ref{thm:seeleyrank}), for each non-adjacent pair 
of vertices $x$ and $y$ there is a value $\alpha_{xy}$ such that for $\alpha\geq\alpha_{xy}$ adding the edge $x\adj y$ is strictly
rank monotone. The proof is completed by taking $\alpha$ larger than all $\alpha_{xy}$'s. 
The result for score monotonicity is similar, using Theorem~\ref{thm:seeleyscore}.
\end{proof}
It is interesting to remark that this result is dual to Theorem~\ref{th:katzzero}: Katz's index is approximated by
degree for values of the damping factor close to the lower bound (zero), whereas PageRank is approximated by degree for values of
the damping factor close to the upper bound (one).

On the other hand, we will now show that 
\emph{for every possible value of the damping factor $\alpha$} there is a graph on which PageRank is neither
rank nor score monotone.
Our proof strategy will be identical to the one we used for Katz's index, except that now we expect our example to
satisfy rank monotonicity when $\alpha$ is close
to its upper bound, instead of its lower bound, because of Theorem~\ref{th:prlim}.

\begin{figure}
\centering
\begin{tabular}{cc}
\raisebox{.5cm}{$G_k$\qquad}&\includegraphics{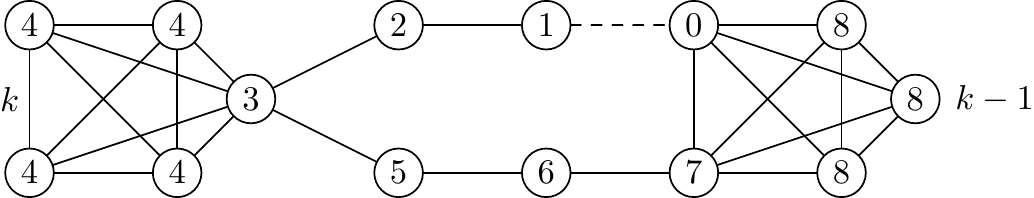}\\
\raisebox{.5cm}{$B_k$\qquad}&\includegraphics{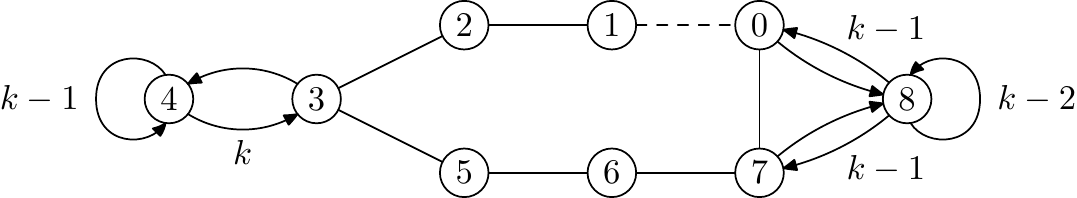}
\end{tabular}
\[
B_k = \left(\begin{matrix}
   0 & 0 & 0 & 0 & 0 & 0 & 0 & \frac1{k} & \frac1{k} \\
    0 & 0 & 1 & 0 & 0 & 0 & 0 & 0 & 0 \\
    0 & \frac12 & 0 & \frac12 & 0 & 0 & 0 & 0 & 0 \\
    0 & 0 & \frac1{k + 2} & 0 & \frac1{k + 2} & \frac1{k + 2} & 0 & 0 & 0 \\
    0 & 0 & 0 & \frac{k}{k} & \frac{k - 1}{k} & 0 & 0 & 0 & 0 \\
    0 & 0 & 0 & \frac12 & 0 & 0 & \frac12 & 0 & 0 \\
    0 & 0 & 0 & 0 & 0 & \frac12 & 0 & \frac12 & 0 \\
    \frac1{k + 1} & 0 & 0 & 0 & 0 & 0 & \frac1{k + 1} & 0 & \frac1{k + 1} \\
    \frac{k - 1}{k} & 0 & 0 & 0 & 0 & 0 & 0 & \frac{k - 1}{k} & \frac{k - 2}{k}
\end{matrix}\right)
\]
\[
B'_k = \left(\begin{matrix}
   0 & \frac1{k + 1} & 0 & 0 & 0 & 0 & 0 & \frac1{k + 1} & \frac1{k + 1} \\
    \frac12 & 0 & \frac12 & 0 & 0 & 0 & 0 & 0 & 0 \\
    0 & \frac12 & 0 & \frac12 & 0 & 0 & 0 & 0 & 0 \\
    0 & 0 & \frac1{k + 2} & 0 & \frac1{k + 2} & \frac1{k + 2} & 0 & 0 & 0 \\
    0 & 0 & 0 & \frac{k}{k} & \frac{k - 1}{k} & 0 & 0 & 0 & 0 \\
    0 & 0 & 0 & \frac12 & 0 & 0 & \frac12 & 0 & 0 \\
    0 & 0 & 0 & 0 & 0 & \frac12 & 0 & \frac12 & 0 \\
    \frac1{k + 1} & 0 & 0 & 0 & 0 & 0 & \frac1{k + 1} & 0 & \frac1{k + 1} \\
    \frac{k - 1}{k} & 0 & 0 & 0 & 0 & 0 & 0 & \frac{k - 1}{k} & \frac{k - 2}{k}
\end{matrix}\right)
\]
\caption{\label{fig:pr}A parametric counterexample graph for PageRank: when adding the edge $0\adj 1$,
vertex $1$ violates score and rank monotonicity (bottom violation). The $k$ vertices labeled with $4$
form a $(k+1)$-clique with vertex $3$, and the $k-1$ vertices labeled with $8$ form a $(k+1)$-clique
with vertices $0$ and $7$. Arc
labels represent multiplicity; weights are induced by the uniform distribution on the upper graph.
The matrices displayed are the adjacency matrix of $B_k$ and $B'_k$; differently from Figure~\ref{fig:eig}, we show them both explicitly to
highlight how the addition of the new edge influences row normalization.}
\end{figure}

In Figure~\ref{fig:pr} we show a family of total graphs $G_k$ depending
on an integer parameter $k$, and an associated family of bases $B_k$, with fibrations defined on the nodes following the node labels,
and on the arcs in the only possible way.\footnote{Note that in the conference version of this paper~\cite{BFVSRMUN}
nodes are numbered differently, and the denominators of the second row of the adjacency matrix
displayed therein are $k-1$, mistakenly, instead of $k+1$.} Weights are defined by
normalizing the adjacency matrix of $G_k$, and then using the fibration to transfer the weights on the arcs $B_k$
(it is easy to see that no conflict arises when multiple arcs of $G_k$ are mapped to the same arc of $B_k$).
As usual, $G'_k$ and $B'_k$ are the same graphs with the additional edge $0\adj 1$.

The basic intuition behind the graphs $G_k$ is that when you connect a high-degree vertex $x$ with a low-degree vertex $y$, 
$y$ will pass to $x$ a much larger fraction of its score than in the opposite direction. This phenomenon is caused by the stochastic normalization of the
adjacency matrix: the arc from $x$ to $y$ will have a low coefficient, due to the high degree of $x$, whereas the arc from $y$ to $x$
will have a high coefficient, due to the low degree of $y$.

While $G_k$ has $2k+6$ vertices, $B_k$ has $9$ vertices, independently of $k$, and
thus its PageRank can be computed analytically as rational functions of $\alpha$ whose coefficients are rational functions
in $k$ (since the number of arcs of each $B_k$ is different).

We thus define, for every node $x$,
\begin{align*}
 \pre_\alpha(x)&= \left[(1-\alpha)\mathbf1\bigl(1 -\alpha  B_k\bigr)^{-1}\right]_x\\
 \post_\alpha(x)&=  \left[(1-\alpha)\mathbf1\bigl(1 -\alpha  B'_k\bigr)^{-1}\right]_x.\\
\end{align*}
Note when discussing score monotonicity we cannot use the adjugate matrix 
to simplify our computations, as we did in Section~\ref{sec:eigen}, but we can use without loss of generality
an arbitrary constant vector as preference vector. When discussing rank monotonicity,
however, we will switch silently to the adjugate (because the denominator cannot change its sign anywhere in $[0\..1)$).

For example,
\[
\post_\alpha(1)= \frac{\frac{2 k^{2} - 6 k + 4}{k^{2} - 6 k + 4} \alpha^{5} + 
\frac{-14 k^{2} + 12 k + 4}{k^{2} - 6 k + 4} \alpha^{4} +\cdots+ 
\frac{6 k^{4} + 6 k^{3} - 24 k^{2} - 24 k}{k^{2} - 6 k + 4} \alpha + 
\frac{-4 k^{4} - 12 k^{3} - 8 k^{2}}{k^{2} - 6 k + 4}}
{\alpha^{5} + \frac{-2 k^{3} - 10 k^{2} + 12 k + 4}{k^{2} - 6 k + 4} \alpha^{4} + \cdots + 
\frac{4 k^{4} + 4 k^{3} - 20 k^{2} - 24 k}{k^{2} - 6 k + 4} \alpha + 
\frac{-4 k^{4} - 12 k^{3} - 8 k^{2}}{k^{2} - 6 k + 4}},
\]
where we omitted part of the terms for lack of space.
Once again, in the adjacency matrix of $B_k$ just four rows contain $k$: as a consequence, the degree in $k$
of numerators and denominators of coefficients of the rational functions in $\alpha$ is at most four.

\subsection{Score monotonicity}

We start by considering node $1$: evaluating $\post_\alpha(1)-\pre_\alpha(1)$ in
$\alpha=2/3$ we obtain a negative value for $k\geq 11$, showing there is
a value of $\alpha$ for which node $1$ violates score monotonicity.
Then, we use again Sturm polynomials to show that
for $k\geq 13$ the numerator of $\post_\alpha(1)-\pre_\alpha(1)$ never changes its sign 
in $\bigl(a_k\..b_k\bigr]$, where
\[
a_k = \frac23 - \frac{2k}{3k+ 100} < \frac23 < \frac23+\frac{k}{3k+100} = b_k,
\]
while the denominator of $\post_\alpha(1)-\pre_\alpha(1)$ cannot have zeros in $[0\..1)$.
The interval $\bigl(a_k\..b_k\bigr]$ approaches $(0\..1]$ as $k$ grows, 
so we conclude that 
the interval of values of $\alpha$ for which the score of node $1$ decreases reaches the whole unit interval as $k$ grows.

Finally, by studying (as in the case of Katz's index) the polynomial $\pre_\alpha(1)-\pre_\alpha(0)$ it is easy to see that in our example
node $0$ is always more important than node $1$ as long as $k\geq 1$.

\begin{thm}
\label{th:prscore}
For every value of $\alpha\in(0\..1)$, for sufficiently large $k$ PageRank with damping factor $\alpha$ is not score monotone (bottom violation) on the
graphs $G_k$ of Figure~\ref{fig:pr}.
\end{thm}

It is also interesting to count the sign changes of $\post_\alpha(1)-\pre_\alpha(1)$ in $\bigl(0\..a_k\bigr]$ (one) and $\bigl(b_k\..1\bigr)$ (one), as they describe the behavior
of the score change for limiting values: initially, the score increases;
then, it starts to decrease somewhere before $a_k$ and stops decreasing somewhere after $b_k$, as expected from Theorem~\ref{th:prlim}.

\subsection{Rank monotonicity}

We now use the same example to prove the lack of rank monotonicity. In this case, we study in a similar way
$\pre_\alpha(1)-\pre_\alpha(5)$, which is positive in $\alpha=2/3$ if $k\geq 13$.
To extend our results about rank monotonicity to every $\alpha$, we use again Sturm polynomials to show that
the numerator of $\pre_\alpha(1)-\pre_\alpha(5)$ never changes its sign in $\bigl(a_k\..b_k\bigr]$ for $k\geq 14$. 

Again, it is interesting to count the sign changes of $p(\alpha)$ in $(0\..a_k]$ (one) and $(b_k\..1)$ (one):
initially, node $1$ has a smaller PageRank than node $5$; then, somewhere before $a_k$,
node $1$ starts having a larger PageRank than $5$; somewhere after $b_k$, we return to the initial condition, as expected from Theorem~\ref{th:prlim}.

Finally, we study $\post_\alpha(1)-\post_\alpha(5)$ which, is negative in $\alpha=2/3$
and again has no sign changes in  $\bigl(a_k\..b_k\bigr]$ for $k\geq 5$. More precisely, we study $p(\alpha)$, where
$\post_\alpha(1)-\post_\alpha(5) = (1-\alpha)^2p(\alpha)$, as $\post_\alpha(1)-\post_\alpha(5)$ is not squarefree, but $p(\alpha)$ is. 

In this case, there are two sign changes in $(0\..a_k]$ and no sign change in $(b_k\..1)$, so
initially, node $1$ is less important than node $5$; then, in an interval of values before $a_k$ it is more important;
then, it starts to be again less important before $a_k$; and it becomes as important as node $5$
only in the limit for $\alpha\to 1$.

\begin{thm}
\label{th:prrank}
For every value of $\alpha\in(0\..1)$, for sufficiently large $k$ PageRank with damping factor $\alpha$ is not rank monotone (bottom violation)
on the graphs $G_k$ of Figure~\ref{fig:pr}.
\end{thm}

Recall that in~\cite{BLVRMCM} PageRank was proven to be both score and
(strictly) rank monotone for all \emph{directed} graphs and all
$\alpha\in[0\..1)$, given that the preference vector is positive: 
comparing those results with Theorems~\ref{th:prscore} and~\ref{th:prrank}, we
see once more that in the undirected case the behavior is radically different.

For sufficiently large $k$, almost all nodes are more important (i.e., have larger PageRank score) than node $1$
both before and after edge addition, with the only exception of nodes $5$ and $6$: as we said, node $5$ is less important than node $1$ before
but more important after the edge addition; whereas node $6$ is also less important than node $1$ before, and becomes as important as node $1$ 
after the edge addition (as node $6$ and node $1$ 
become equivalent modulo an automorphism). As a result, node $1$ is demoted.

Finally, we provide in Figure~\ref{fig:pr-new} a counterexample in which the more important node violates rank monotonicity.
In this case, the intuition is that we connect two nodes with the same degree but different scores.
As in the previous case, the counterexample works for any chosen $\alpha$, up to an appropriate choice of the parameter $k$.
The proof follows the same line of attack, and detailed computations can be found in the Sage worksheets. The main
difference is that the relevant interval $\bigl(a_k\..b_k\bigr]$ is now
\[
a_k = \frac23-\frac{2\sqrt k}{3\sqrt k+10}<\alpha\leq \frac23 + \frac k{3k+1} = b_k.
\]

\begin{thm}
\label{th:prrank-new}
For every value of $\alpha\in(0\..1)$, for sufficiently large $k$ PageRank with damping factor $\alpha$ is neither score nor rank monotone (top violation)
on the graphs $G_k$ of Figure~\ref{fig:pr-new}.
\end{thm}

\begin{figure}
\centering
\begin{tabular}{cc}
\raisebox{1cm}{$G_k$\qquad}&\includegraphics{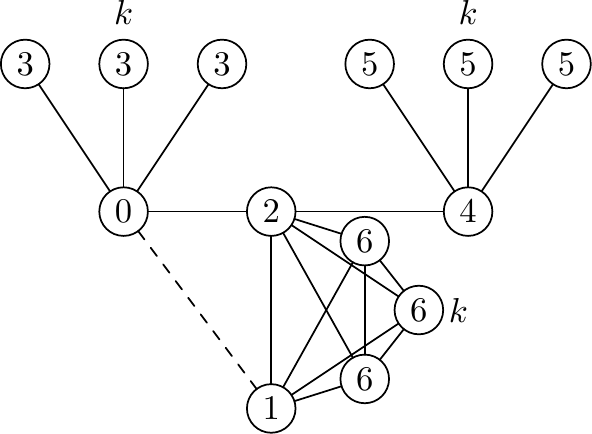}\\
\raisebox{1cm}{$B_k$\qquad}&\includegraphics{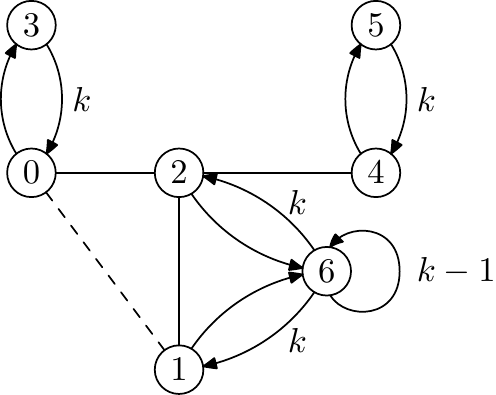}
\end{tabular}
\caption{\label{fig:pr-new}
A parametric counterexample graph for PageRank: when adding the edge $0\adj 1$,
vertex $0$ violates score and rank monotonicity (top violation).
There is a star with $k$ leaves around vertex $0$, a star with $k$ leaves around vertex $4$,
and the $k$ vertices labeled with $6$ form a $(k+2)$-clique with vertices $1$ and $2$.  
Arc labels represent multiplicity; weights are induced by the uniform distribution on the upper graph.}
\end{figure}

\subsection{Anecdotal Evidence: PageRank on the IMDB}

\begin{table}[t]
\renewcommand{\arraystretch}{1.2}
\renewcommand{\tabcolsep}{1ex}
\begin{tabular}{lll}
Score increase & Score decrease & Violations of rank monotonicity\\
\hline
Meryl Streep & Yasuhiro Tsushima & Anne--Mary Brown, Jill Corso,~\ldots\\
Denzel Washington & Corrie Glass & Patrice Fombelle, John Neiderhauser,~\ldots\\
Sharon Stone & Mary Margaret (V) & Dolores Edwards, Colette Hamilton,~\ldots\\
John Newcomb & Robert Kirkham & Brandon Matsui, Evis Trebicka,~\ldots 
\end{tabular}
\vspace*{.5em}
\caption{\label{tab:rank}A few examples of violations of score monotonicity and rank monotonicity in the Hollywood co-starship graph
\texttt{hollywood-2011}. If we add an edge between the actors in the first and second column, 
the first actor has a score increase, the second actor has a score decrease,
and the actors in the third column, which were less important than the second actor, become more important after the edge addition.
The first three examples are bottom violations, whereas the last one is a top violation.}
\end{table}

To show that our results are not only theoretical, we provide a few interesting anecdotal examples from
the PageRank scores ($\alpha=0.85$) of the Hollywood co-starship graph,
whose vertices are actors/actresses in the Internet Movie Database, with an edge connecting them if played in the same movie.
In particular, we used the \texttt{hollywood-2011} dataset from the Laboratory for Web Algorithmics,\footnote{\url{http://law.di.unimi.it/}}
which contains approximately two million vertices and $230$ million edges.

To generate our examples, we picked two actors either at random, or considering
the top $1/10000$ of the actors of the graph in PageRank order and the bottom
quartile, looking for a collaboration that would hurt either actor (or
both).\footnote{Note that for this to happen, the collaboration should be a
two-person production. A production with more people would add more
edges.} About $4$\% of our samples yielded a violation of monotonicity, and in
Table~\ref{tab:rank} we report a few funny examples.

The first three cases are bottom violations: it is the less-known actor that loses score (and rank) by the collaboration
with the star, and not the other way round, as it happens also in the counterexample of Figure~\ref{fig:pr}.
In the last case, instead, we hava top violation: 
a collaboration would damage the more important vertex, like in the counterexample of Figure~\ref{fig:pr-new}.
We found no case in which both actors would be hurt by the collaboration, and it is in fact an open problem whether this situation 
can happen.

\section{Conclusions}

We have studied score and rank monotonicity on undirected graphs for some popular notions of centrality.
Our results show that except for Seeley's index (on connected graphs) there are always cases in which rank monotonicity 
does not hold, and in the case of Katz's index and PageRank we can find range of values of the parameters where the
violation occurs;
moreover, some centralities are also not score monotone. We provide examples of both top and bottom violations
to highlight that even the knowledge of whether one is the more important or less important node is insufficient 
to decide whether the new edge will be beneficial. A possible direction for future research
is to show that top and bottom violations cannot happen at the same time, that is, that the new
edge is beneficial for at least one endpoint. 

This lack of monotonicity is opposite to that we observed in the directed case, and it can also be seen in real-world graphs (at least for PageRank).
It is interesting to note that even centrality indices that were designed for undirected graphs (e.g., closeness) are 
not rank monotone in the undirected case (even under a connectedness assumption).
Our results show that common knowledge and intuitions about the behavior of centrality measures in the directed case cannot
be applied to the undirected case.

\bibliography{biblio}

\end{document}